\documentclass[conference, final]{IEEEtran} 
\usepackage{amsmath, amsthm, amssymb}
\usepackage{epsfig}
\usepackage{latexsym}
\usepackage{graphicx}
\usepackage{cite}
\usepackage{subfigure}

\newtheorem{theorem}{Theorem}

\newtheorem{corollary}{Corollary}

\title{ The Area Under a Receiver Operating Characteristic Curve Over Enriched Multipath Fading Conditions}

\author{
Paschalis~C.~Sofotasios$^{1,2}$,~Mulugeta~K.~Fikadu$^{1}$,~Khuong~Ho-Van$^{3}$,~Mikko~Valkama$^{1}$~and~George~K.~Karagiannidis$^{2,4}$
\\\\
\begin{normalsize} 
$^{1}$Department of Electronics and Communications Engineering, Tampere University of Technology, 33101 Tampere, Finland. 
\end{normalsize}
\\
\begin{normalsize} 
e-mail: $\rm \left\lbrace paschalis.sofotasios; mulugeta.fikadu; mikko.e.evalkama \right\rbrace@\rm tut.fi$ 
\end{normalsize} 
\\
\begin{normalsize} 
$^{2}$Department of Electrical and Computer Engineering, Aristotle University of Thessaloniki, 54124 Thessaloniki, Greece.
\end{normalsize}
\\
\begin{normalsize} 
e-mail:  $ \rm geokarag@\rm ieee.org$
\end{normalsize} 
\\
\begin{normalsize} 
$^{3}$Department of Electrical Engineering, HoChiMinh City University of Technology, HoChiMinh City, Vietnam. 
\end{normalsize}
\\
\begin{normalsize} 
e-mail: $\rm hvkhuong@\rm hcmut.edu.vn$
\end{normalsize} 
\\
\begin{normalsize} 
$^{4}$Department of Electrical and Computer Engineering, Khalifa University, PO Box 127788, Abu Dhabi, UAE. 
\end{normalsize} 
 }

\begin{document}
\maketitle

\begin{abstract}
This work is devoted to the analysis of the performance  of energy detection based spectrum sensing in the presence of enriched fading conditions which are distinct for the large number of multipath components and the lack of a dominant components. This type of fading conditions are characterized efficiently by the well known Nakagami${-}q$ or  Hoyt distribution and the proposed analysis is  carried out in  the context of the area under the receiver operating characteristics (ROC) curve (AUC). Unlike the widely used probability of detection metric, the AUC is a single metric and has been shown to be rather capable of evaluating the performance of a detector in applications relating to cognitive radio, radar systems and biomedical engineering, among others. Based on this, novel analytic expressions are derived for the average AUC and its complementary metric, average CAUC, for both integer and fractional values of the involved time-bandwidth product. The derived expressions have a tractable algebraic representation which renders them convenient to handle both analytically and numerically. Based on this, they are employed in analyzing the behavior of  energy detection based spectrum sensing over enriched fading conditions for different severity scenarios, which demonstrates that the performance of energy detectors is, as expected, closely related to the value of the fading parameter $q$. 
\end{abstract}

\section{Introduction}

The detection of unknown signals has been an  important research topic  over the past decades. This has been largely required in the context of radar systems and more recently in cognitive radio based communications which have attracted a significant interest by both academia and industry due to the prominent ability to increase the utilization of the currently scarce spectrum resources. The operational principle of detecting unknown signals is typically based on spectrum sensing (SS) with energy detection (ED) constituting the most simple and widely adopted method. The ED is based on the deployment of a radiometer, which is a non-coherent detection device that measures the energy level of a received signal waveform over an observation time window. The obtained measure is then compared to a pre-defined energy threshold and based on this it is determined whether an unknown signal is present or absent \cite{J:Marcum_1, J:Swerling, B:Haykin, J:Haykin, B:Bargava}.

Detection of unknown signals  over a flat band-limited Gaussian noise channel was firstly addressed  in \cite{J:Urkowitz} where  analytic expressions were proposed for the   probability of detection, $P_{d}$,  and probability of false alarm, $P_{f}$, measures. These performance metrics are  based on the statistical assumption that the decision statistics follow the  central chi-square and the non-central chi-square distributions, respectively. A few decades later, this problem was revisited in  \cite{C:Kostylev, J:Alouini} assuming quasi-deterministic signals over fading channels.  To this effect, numerous studies analyzed the performance of energy detectors under different communication and fading scenarios. Specifically, the authors in \cite{J:Alouini} derived closed-form expressions for the average probability of detection over Rayleigh, Rice and Nakagami${-}m$ fading channels for both single-channel and multi-channel scenarios.  The ED performance in the case of equal gain combining over Nakagami${-}m$ multipath fading was reported in \cite{C:Herath} while the performance in collaborative spectrum sensing and in relay-based cognitive radio networks was thoroughly investigated in \cite{C:Ghasemi, C:Sousa, J:Ghasemi-Sousa, C:Attapattu, J:Attapattu_2}. In the same context, a semi-analytic method for analyzing the performance of energy detection of unknown deterministic signals was reported in \cite{J:Herath} and is based on the moment-generating function (MGF) method. This method was utilized in diversity methods in the presence of Rayleigh, Rice and Nakagami${-}m$ fading in \cite{J:Herath} as well as  for the useful case of correlated Rayleigh and Rician fading channels in \cite{C:Beaulieu}. Finally, the detection of unknown signals in low signal-to-noise-ratio (SNR) over $K{-}$distributed ($K$), generalized $K$ ($K_{G}$) as well as in generalized multipath fading channels was analyzed in \cite{C:Atapattu_3, J:Janti, J:Attapattu, New_1, C:Attapattu_2, J:Sofotasios, New_2, C:Tellambura_2, New_3, New_3a, New_3b, New_4, J:Ghasemi, New_5, New_6}.   

It is noted here that nearly all analyses and investigations on the behavior of spectrum sensing techniques are based exclusively on the probability of detection and the  probability of false alarm performance measures. However, another performance measure that is capable of evaluating adequately the performance of energy detectors exists, the so called \textit{area under the receiver operating characteristic curve} (ROC). This performance measure is denoted as  AUC and it has been widely used in several disciplines of sciences and engineering \cite{New_7, New_8} and the reference therein. The distinct feature of this measure is that, contrary to $P_{d}$, it constitutes a single parameter measure and thus, it accounts for the overall performance of the detector in a more general manner. For example, the $P_{d}$ measure is expressed as: $i)$ a function of the instantaneous SNR,  $\gamma$, for the unfaded case  or  the average SNR, $\overline{\gamma}$, for the faded case with fixed values of $P_{f}$; $ii)$  a function of $P_{f}$ with fixed values of $\gamma$ or $\overline{\gamma}$. Contrary to that, the AUC is expressed as a function of $\gamma$ for the unfaded case or $\overline{\gamma}$ for the faded case. Therefore, the AUC  provides a more general insight  on the overall behavior and performance of a detector as a function of $\gamma$ and $\overline{\gamma}$.  Based on this, the authors in \cite{Tellambura_AUC_3, Tellambura_AUC_2, J:Tellambura_AUC, Annamalai_2, Annamalai_3}  addressed the AUC under different communication and fading scenarios. In more details, the authors in \cite{Tellambura_AUC_3, J:Tellambura_AUC, Annamalai_3} analyzed AUC in different diversity schemes and multipath fading channels. Likewise, investigation of AUC   in  amplify-and-forward relay scenarios were given in \cite{Annamalai_2} while the complementary area under the ROC curve measure was proposed in \cite{Tellambura_AUC_2}. 

It is also widely known that the last decades witnessed numerous advances in wireless channel characterization and modeling. In this context, the Nakagami${-}q$ or Hoyt distribution  has been shown to be an accurate fading model for accounting for various multipath fading scenarios. The distinct feature of Hoyt fading model is that it is capable of accounting for enriched multipath fading conditions in  the absence of any dominant component \cite{B:Alouini, Yacoub_1, Zogas, Iskander, Matalgah, J:Paris_1, J:Tavares, J:Paris_2, J:Beaulieu, Yacoub_2} - and the references therein.  Nevertheless, in spite of its usefulness it has not been explicitly considered in the context of energy detection. Motivated by this, the present work is devoted in analyzing the effect of Hoyt distributed multipath fading in the performance of energy detectors. To this end, novel analytic expressions are derived for the corresponding AUC and CAUC measures for both integer and fractional values of the involved time-bandwidth product. These expressions have a relatively simple algebraic representation which constitutes them relatively convenient to handle both analytically and numerically. With the aid of these expressions, it is shown that the performance of energy detection based spectrum sensing is rather sensitive at the severity of the fading conditions. This is particularly the case for moderate   SNR  values as even slight variations of $q$ create
 an immediate effect on the value of AUC and CAUC.

The remainder of this paper is organized as follows: The system and channel model are described in Section II. The area under the ROC curve and complementary area under the ROC curve of the energy detector over Hoyt fading channels are analyzed in Section III.  Numerical results for various communication scenarios and discussions are provided in Section IV while closing remarks are given in Section V.

\section{System and Channel Model}

\subsection{Energy Detection-Based Spectrum Sensing}

The received signal waveform in narrowband energy detection follows the standard binary hypothesis \cite[eq. (1)]{C:Beaulieu} (and the references therein),
\begin{equation} \label{Test_1} 
r(t) =
\begin{cases}
n(t) \,\,\,\, \qquad \,\,\,\, \qquad \,\,\, \,\,\,\,\,\,\,\,\,\,\,\,\,\,\,\,\,\,:H_{0} \\ 
hs(t) + n(t) \, \qquad \, \,\,\,\,\,\,\,\,\,\,\,\,\,\,\,\,\,\,\,:H_{1} 
\end{cases}
\end{equation}
where $s(t)$ is an unknown deterministic signal and  $h$, $n(t)$ denote the complex gain of the channel coefficient and   an additive white Gaussian noise (AWGN) process, respectively. The samples of $n(t)$ are assumed to be zero-mean Gaussian random variables with variance $N_{0}W$ with $W$ and $N_{0}$ denoting the  single-sided signal bandwidth and a single-sided noise power spectral density, respectively  \cite{C:Beaulieu}. The hypothesis $H_{1}$ refers to the case that an information signal is present whereas $H_{0}$ refers to the case that an   information signal is absent.  The received signal is subsequently filtered, squared and integrated over the time interval $T$ in \cite[eq. (2)]{J:Alouini}, namely, 
\begin{equation}
y \triangleq \frac{2}{N_{0}} \int_{0}^{T} \mid r(t)\mid ^{2} dt. 
\end{equation} 
The output of the integrator is a measure of the energy of the received waveform which constitutes a test statistic that determines whether the received energy measure corresponds only to the energy of noise ($H_{0}$) or to the energy of both the unknown deterministic signal and noise ($H_{1}$). By denoting the observation time bandwidth product as $u = TW$,  the test statistic follows the central chi-square distribution with $2u$ degrees of freedom under the $H_{0}$ hypothesis and the non central chi-square distribution with $2u$ degrees of freedom under the $H_{1}$ hypothesis \cite{J:Urkowitz}. To this effect, the corresponding probability density function (PDF) in the presence of AWGN is expressed according to \cite[eq. (3)]{J:Alouini}, namely, 
\begin{equation} \label{Test_3} 
p_{Y}(y) = 
\begin{cases}
\frac{1}{2^{u}\Gamma(u)}y^{u-1}e^{-\frac{y}{2}} \,\, \qquad \,\, \qquad \,\, \qquad \,\,\,\,\, \,\,\,\,:H_{0} \\ 
\frac{1}{2}\left(\frac{y}{2\gamma} \right)^{\frac{u-1}{2}} e^{-\frac{y + 2\gamma}{2}}I_{u-1}\left(\sqrt{2y\gamma} \right) \, \, \,\,\,\,\,\,\,\,:H_{1}
\end{cases}
\end{equation}
where $\gamma \triangleq   |h| ^{2} E_{s}{/}N_{0}$ is the corresponding instantaneous SNR, $E_{s}$ is the signal energy and $\Gamma\left(. \right)$, $I_{n}\left( . \right)$ denote the  gamma function and the modified Bessel function of the first kind, respectively, \cite{B:Prudnikov, B:Tables_1, B:Tables_2, B:Ryzhik}.

As already mentioned, an energy detector is largely characterized by a predefined energy threshold, $\lambda$. This threshold is particularly critical in the decision process and is promptly associated to three measures that overall evaluate the performance of the detector: i) the probability of false alarm, $P_{f}=Pr(y> \lambda \mid H_{0})$;  ii) the probability of detection, $P_{d}=Pr(y> \lambda \mid H_{1})$ and iii) the probability of missed detection, $P_{m} = 1 - P_{d}$. The first two measures are deduced by integrating \eqref{Test_3} over the interval between the energy threshold to infinity, $ \left\lbrace   \lambda, \, \infty  \right\rbrace $, yielding \cite{J:Alouini}, 
\begin{equation} \label{Pf_1} 
P_{f} = \frac{\Gamma \left(u,\frac{\lambda}{2}\right)}{\Gamma(u)} 
\end{equation}
\noindent 
and
\begin{equation}\label{Pd_1} 
P_{d} = Q_{u}(\sqrt{2 \gamma},\sqrt{\lambda})
\end{equation}
\noindent 
where $Q_{m}(a,b)$ and $\Gamma\left(., .\right)$   denote the  generalized Marcum${-}Q$ function and the upper incomplete gamma function, respectively \cite{J:Marcum_1, B:Ryzhik}.

\subsection{AUC and CAUC Performance Metrics }

The area under the ROC curve is a single parameter measure that can be used to evaluate the overall performance of energy detectors and therefore of spectrum sensing in cognitive radio and radar systems.  The AUC is defined as,
\begin{equation} \label{AUC_1}
{\rm AUC} = A(\gamma) \triangleq \int_{0}^{1} P_{d}(\gamma, \lambda) dP_{f}(\lambda)
\end{equation}
and can be equivalently expressed as follows,
\begin{equation} \label{AUC_2} 
{\rm AUC} = A(\gamma) = - \int_{0}^{\infty} P_{d}(\gamma, \lambda) \frac{\partial P_{f}(\lambda)}{\partial \lambda} d \lambda. 
\end{equation}
For the case that $P_{f}(\lambda)$ and $P_{d}(\gamma, \lambda)$ follow the central chi-square  and the non central chi-square distributions, respectively, the AUC can be expressed by \cite[eq. (9)]{Tellambura_AUC_3}, namely, 
\begin{equation} \label{AUC_3}
A(\gamma) = 1 - \sum_{l = 0}^{u - 1} \frac{\gamma^{l} e^{- \frac{\gamma}{2}}}{l! 2^{l}} + \sum_{l = 1 - u}^{u - 1} \frac{(u)_{l}\,_{1}F_{1}\left( u + l, 1 + l, \frac{\gamma}{2} \right)}{l! \, 2^{u + l } }
\end{equation}
where 
\begin{equation*}
(a)_{n} \triangleq \frac{\Gamma(a + n)}{\Gamma(a)}
\end{equation*}
denotes  the Pochhammer symbol \cite{B:Tables_1}. In the same context, the complementary area under the ROC curve is given by, 
\begin{equation} \label{CAUC_1}
{\rm CAUC} = A'(\gamma) = \int_{1}^{0} P_{m}(\gamma, \lambda) dP_{f}(\lambda).  
\end{equation}
Equation \eqref{CAUC_1} can be re-written as follows, 
\begin{equation}
\int_{1}^{0} P_{m}(\gamma, \lambda) dP_{f}(\lambda) = - \int_{0}^{1} P_{m}(\gamma, \lambda) dP_{f}(\lambda)
\end{equation}
which can be equivalently expressed as, 
\begin{equation}
  \int_{0}^{1} \left[ 1 - P_{d}(\gamma, \lambda) \right]  dP_{f}(\lambda)  = \int_{0}^{1}  dP_{f}(\lambda) -    \int_{0}^{1}   P_{d}(\gamma, \lambda)  dP_{f}(\lambda) 
\end{equation}
Based on this and with the aid of \eqref{AUC_1} and \eqref{AUC_2}, it follows that $ {\rm CAUC} = 1 - {\rm AUC}$, i.e. 
\begin{equation} \label{CAUC_2}
A'(\gamma) = 1 - A(\gamma). 
\end{equation}
 A closed-form expression for CAUC for integer values of $u$ is deduced by substituting \eqref{AUC_3} in \eqref{CAUC_2} yielding \cite[eq. (7)]{Tellambura_AUC_2},
\begin{equation} \label{CAUC_3}
A'(\gamma) =  \sum_{l = 0}^{u - 1} \frac{\gamma^{l} e^{- \frac{\gamma}{2}}}{l! 2^{l}}  - \sum_{l = 1 - u}^{u - 1} \frac{(u)_{l}\,_{1}F_{1}\left( u + l, 1 + l, \frac{\gamma}{2} \right)}{l! \, 2^{u + l } }. 
\end{equation}
It is recalled that  \eqref{AUC_3} and \eqref{CAUC_3} hold  for integer values of $u$. 

\subsection{The Hoyt (Nakagami${-}q$) Fading Distribution}

As already mentioned, the Hoyt  fading model has been shown to represent effectively the small-scale variations of an information signal in non-line-of-sight (NLOS) communication scenarios. Physically, this fading model  accounts for the case of enriched multipath fading with no presence of a dominant component. The  PDF of the instantaneous SNR in Hoyt fading channels is given by  \cite[eq. (2.11)]{B:Alouini}, namely, 
\begin{equation} \label{PDF_Hoyt} 
p_{\gamma}(\gamma) = \frac{1 + q^{2}}{2q \overline{\gamma}} e^{- \frac{(1 + q^{2})^{2} \gamma}{4 q^{2}\overline{\gamma}}} I_{0}\left( \frac{(1 - q^{4}) \gamma}{4 q^{2} \overline{\gamma}} \right) 
\end{equation}
where $\overline{\gamma}$ is the average SNR and $q$ denotes the   Nakagami${-}q$ fading parameter, which  is valid in the range $0 < q \leq 1$.  Likewise, the CDF  of the Hoyt model is given by \cite[eq. (9)]{J:Paris_1}
\begin{equation} \label{CDF}
\begin{split}
P_{\gamma}(\gamma) &= Q_{1}\left( \sqrt{ \frac{(1 - q^{4}) (1 - q) \gamma}{8q (1+q) \overline{\gamma} }}, \sqrt{\frac{(1 + q^{4}) (1 + q) \gamma}{8q (1-q) \overline{\gamma} }}\right)   \\
&- Q_{1}\left( \sqrt{\frac{(1 + q^{4}) (1 + q) \gamma}{8q (1-q) \overline{\gamma} }}, \sqrt{\frac{(1 - q^{4}) (1 - q) \gamma}{8q (1+q) \overline{\gamma} }}\right) 
\end{split}
\end{equation}
and the corresponding MGF is expressed as \cite[eq. (2.12)]{B:Alouini},
\begin{equation} \label{MGF_Hoyt}
\begin{split}
M_{\gamma}(s)  = \frac{1}{\sqrt{1 - 2s \overline{\gamma} + \left(  \frac{2 s \overline{\gamma} q}{1 + q^{2}} \right)^{2} }}  
\end{split}
\end{equation} 
in a simple form that involves only elementary function. 

\section{AUC and CAUC over Hoyt Fading Channels }

\subsection{AUC for the Special Case that $u$ is a Real Positive Integer}

\begin{theorem}
For $\overline{\gamma} \in \mathbb{R}^{+}$, $u \in \mathbb{N}$ and $0 < q \leq 1$, the following closed-form expression is valid for the AUC over Hoyt (Nakagami${-}q$) fading channels,  
\begin{equation} \label{AUC_Hoyt_1}
\overline{A}_{\rm Hoyt} = 1 - \sum_{l = 0}^{u - 1} \sum_{i = 0}^{l} \binom{l + u - 1}{l - i} \frac{q^{1 + 2i} 2^{i + 1 -  l - u} \, \overline{ \gamma}^{\, i}  }{\left(2 \overline{\gamma} q^{2} + (1 + q^{2})^{2} \right)^{i + 1 } }  \qquad \qquad \qquad
   \end{equation}
 \begin{equation*}
 \quad  \quad \qquad   \times \,_{2}F_{1}\left( \frac{i + 1}{2}, \frac{i}{2} + 1; 1; \frac{(1 - q^{4})^{2}}{\left[  (1 + q^{2})^{2}+ 2 q^{2} \overline{\gamma}\,  \right]^{2} }  \right)
\end{equation*}
where 
\begin{equation}
\binom{a}{b} \triangleq \frac{a!}{b! (a-b)!}
\end{equation}
denotes the binomial coefficient and  
\begin{equation}
\, _{2}F_{1}(a, b; c; x)  \triangleq = \sum_{l = 0}^{\infty} \frac{(a)_{l} (b)_{l}}{(c)_{l}} \frac{x^{l}}{l!}
\end{equation}
is the Gauss hypergeometric function \cite{B:Tables_1, B:Tables_2, B:Sofotasios, B:Ryzhik}.  
\end{theorem}

\begin{proof}
When $u \in \mathbb{N}$, the AUC can be expressed according to \eqref{AUC_3}. Importantly, it can be also expressed in terms of the generalized Laguerre polynomial, $L_{n}^{a}(x)$, as follows  \cite{Annamalai_2}, 
\begin{equation} \label{AUC_Hoyt_2}
A(\gamma) = 1 - \sum_{l = 0}^{u - 1} \frac{L^{u}_{l} \left( - \frac{\gamma}{2} \right)}{2^{l + u} e^{\frac{\gamma}{2}} }. 
\end{equation}
As a result, the average AUC over fading channels can be expressed as,
\begin{equation} \label{AUC_Hoyt_3a}
\overline{A} = 1 - \sum_{l = 0}^{u - 1} \frac{1}{2^{l + u}} \int_{0}^{\infty} L_{l}^{u} \left(- \frac{\gamma}{2} \right) e^{- \frac{\gamma}{2}} p_{\gamma}(\gamma) d \gamma 
\end{equation}
where $p_{\gamma}(\gamma)$ denotes the SNR PDF of the fading statistics and $\int_{0}^{\infty} p_{\gamma}(\gamma) d \gamma \triangleq  1$. For the case of Hoyt fading channels, one needs to substitute \eqref{PDF_Hoyt} in \eqref{AUC_Hoyt_3a}. To this effect and by expanding $L_{n}^{a}(x)$  according to \cite[eq. (8.970.1)]{B:Ryzhik} one obtains,  
\begin{equation} \label{AUC_Hoyt_3b}
\overline{A}_{\rm Hoyt}  = 1  -  \sum_{l = 0}^{u - 1}  \sum_{i = 0}^{l}  \binom{l + u - 1}{l - i}  \frac{1 + q^{2}}{i! q \overline{\gamma} 2^{l + i + u + 1}}  \qquad \qquad \qquad 
\end{equation} 
\begin{equation*}
\qquad \qquad  \quad \times  \int_{0}^{\infty}    \gamma^{i}  e^{- \gamma \left( \frac{(1 + q^{2})^{2}  }{4 q^{2}\overline{\gamma} }+ \frac{1}{2} \right)}    I_{0}\left( \frac{(1 - q^{4}) \gamma}{4 q^{2} \overline{\gamma}} \right)   d \gamma.  
\end{equation*}
Notably, the class of  $\int_{0}^{\infty} x^{a} {\rm exp}(-bx)I_{0}(cx) dx$ integrals can be expressed in closed-form with the aid of \cite[eq. (6.621.1)]{B:Ryzhik} and \cite[eq. (8.406.3)]{B:Ryzhik}. Therefore, by performing the necessary change of variables and substituting in \eqref{AUC_Hoyt_3b}, equation  \eqref{AUC_Hoyt_1} is deduced, which completes the proof. 
\end{proof}

\subsection{AUC for the Case that $u$ is an Arbitrary Positive Real}

\begin{theorem}
For $\overline{\gamma}, u \in \mathbb{R}^{+}$  and $0 < q \leq 1$, the following closed-form expression is valid for the AUC over Hoyt (Nakagami${-}q$) fading channels,  
\begin{equation} \label{AUC_Hoyt_4}
\begin{split}
\overline{A}_{\rm Hoyt}  = & \sum_{l = 0}^{\infty} \frac{q^{1 + 2l}   (1 + q^{2}) \overline{\gamma} \, (l + u)_{u} \,_{2}F_{1} \left( 1, l + 2u; 1 + u; \frac{1}{2} \right) }{u! 2^{2u - l - 1} \left( 4 \overline{\gamma} q^{2} + (1 + q^{2})^{2} \right)^{l + 1}  }   \\
& \times \, _{2}F_{1}\left( \frac{l + 1}{2}, \frac{l}{2} + 1; 1; \frac{(1 - q^{4})^{2}}{\left(4 \overline{\gamma} q^{2} + (1 + q^{2})^{2} \right)^{2}} \right). 
\end{split}
\end{equation}
\end{theorem}

\begin{proof}
When the value of the time-bandwidth product is   arbitrary real, the AUC can be expressed as follows \cite{Annamalai_3},
\begin{equation} \label{AUC_Hoyt_5}
A(\gamma) = \sum_{l = 0}^{\infty} \frac{ \gamma^{i} ( l + u)_{u} }{l! u! 2^{l + 2u} e^{\gamma} } \, _{2}F_{1}\left( 1, l + 2u; 1 + u; \frac{1}{2} \right). 
\end{equation}
Thus, by  averaging \eqref{AUC_Hoyt_5} over the fading statistics in \eqref{PDF_Hoyt} yields, 
\begin{equation} \label{AUC_Hoyt_6}
\overline{A}_{\rm Hoyt}  =\sum_{l = 0}^{\infty} \frac{ ( l + u)_{u} (1 + q^{2}) }{l! u! 2^{l + 2u + 1} q \overline{\gamma}   } \, _{2}F_{1}\left( 1, l + 2u; 1 + u; \frac{1}{2} \right)   \qquad \qquad \qquad 
\end{equation}
\begin{equation*}
\quad \times \int_{0}^{\infty}  \gamma^{i}  e^{- \gamma \left(1 +  \frac{(1 + q^{2})^{2}  }{4 q^{2}\overline{\gamma}}\right)} I_{0}\left( \frac{(1 - q^{4}) \gamma}{4 q^{2} \overline{\gamma}} \right)  d \gamma.  
\end{equation*}
The integral in \eqref{AUC_Hoyt_6} has the same algebraic form as the integral in \eqref{AUC_Hoyt_3b}. Therefore, by making the necessary change of variables in \cite[eq. (6.621.1)]{B:Ryzhik} and \cite[eq. (8.406.3)]{B:Ryzhik}, substituting in \eqref{AUC_Hoyt_6} and carrying out some algebraic manipulations, one obtains  \eqref{AUC_Hoyt_4} thus, completing the proof. 
\end{proof}

\subsection{CAUC for Positive Integer and Arbitrary Real Values of $u$}

As already mentioned, the CAUC is a complementary measure to the AUC.

\begin{corollary}
For   $\overline{\gamma}$, $0 < q \leq 1$ and for the cases that $u \in \mathbb{N}$, the following expressions hold for the CAUC over Hoyt (Nakagami${-}q$) fading channels,  
\begin{equation} \label{CAUC_Hoyt_1}
\overline{A'}_{\rm Hoyt} =  \sum_{l = 0}^{u - 1} \sum_{i = 0}^{l} \binom{l + u - 1}{l - i} \frac{q^{1 + 2i} 2^{i + 1 -  l - u} \, \overline{ \gamma}^{\, i}  }{\left(2 \overline{\gamma} q^{2} + (1 + q^{2})^{2} \right)^{i + 1 } }  \qquad \qquad \qquad
   \end{equation}
 \begin{equation*}
\qquad \qquad  \quad \times \,_{2}F_{1}\left( \frac{i + 1}{2}, \frac{i}{2} + 1; 1; \frac{(1 - q^{4})^{2}}{\left[  (1 + q^{2})^{2}+ 2 q^{2} \overline{\gamma}\,  \right]^{2} }  \right)
\end{equation*}
\end{corollary}

\begin{proof}
Given that $A'(\gamma) = 1 - A(\gamma)$,  it follows that $\overline{A'}_{\rm Hoyt} = 1 - \overline{A'}_{\rm Hoyt}$. Based on this, the proof is immediately completed by   substituting   \eqref{AUC_Hoyt_1} and \eqref{AUC_Hoyt_4}.  
\end{proof}
In the same context, an analytic expression can be derived for the case of arbitrary value of $u$. 
\begin{corollary}
For   $\overline{\gamma}$, $0 < q \leq 1$ and for the cases that  $ u \in \mathbb{R}^{+}$, the following expressions hold for the CAUC over Hoyt (Nakagami${-}q$) fading channels,  
\begin{equation} \label{CAUC_Hoyt_2}
\overline{A'}_{\rm Hoyt}  =  1 -  \sum_{l = 0}^{\infty} \frac{  (1 + q^{2}) \overline{\gamma} \, (l + u)_{u} \,_{2}F_{1} \left( 1, l + 2u; 1 + u; \frac{1}{2} \right) }{u! 2^{2u - l - 1} q^{- 1 - 2l}   \left( 4 \overline{\gamma} q^{2} + (1 + q^{2})^{2} \right)^{l + 1}  }    \qquad \qquad \qquad 
\end{equation}
\begin{equation*}  
\qquad  \qquad  \quad  \times \, _{2}F_{1}\left( \frac{l + 1}{2}, \frac{l}{2} + 1; 1; \frac{(1 - q^{4})^{2}}{\left(4 \overline{\gamma} q^{2} + (1 + q^{2})^{2} \right)^{2} } \right). 
\end{equation*} 
\end{corollary}

\begin{proof}
The proof follows immediately by Corollary $1$. 
\end{proof}
To the best of the authors knowledge,  equations \eqref{AUC_Hoyt_1}, \eqref{AUC_Hoyt_4},  \eqref{CAUC_Hoyt_1} and \eqref{CAUC_Hoyt_2} have not been previously reported in the open technical literature.

\section{Numerical Results} 

Having derived new analytic expressions for  the AUC and CAUC measures, this section is devoted to  the  analysis 
of the respective behaviour of energy detection-based spectrum sensing  over Hoyt fading channels. The corresponding performance is evaluated for different communication scenarios by means of  $\overline{A}_{\rm Hoyt}$ vs $\overline{\gamma}$ and  $\overline{A'}_{\rm Hoyt}$ vs $\overline{\gamma}$  curves.   To this end, 
\begin{figure}[h!]
\includegraphics[ width=9.25cm,height=6.475cm]{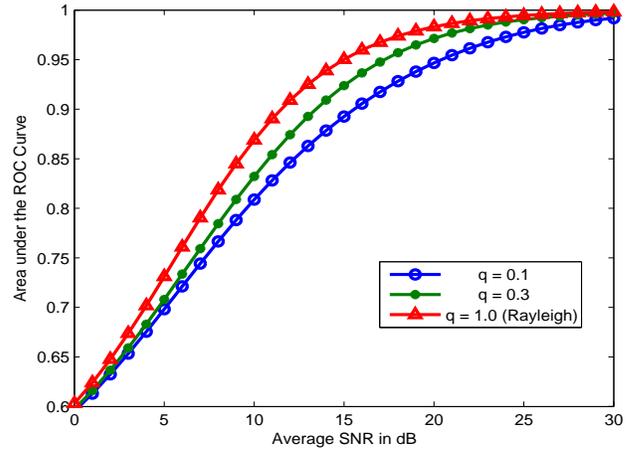} 
\caption{ $\overline{A}_{\rm Hoyt}$ vs $\overline{\gamma}$ for i.i.d Hoyt (Nakagami${-}q$) fading channels with $u = 5$ and different values of $q$. } 
\end{figure}
 Fig. $1$ demonstrates the average AUC  as a function of the average SNR for $u = 5$ and different values of the Nakagami${-}q$ fading parameter. The special case $q  = 1$, which corresponds to Rayleigh fading, is also depicted for comparison. Evidently, the performance of the detector is nearly excellent in the high SNR regime $\overline{\gamma} > 25$dB irrespective of the severity of fading. 
 \begin{figure}[h!]
\includegraphics[ width=9.25cm,height=6.475cm]{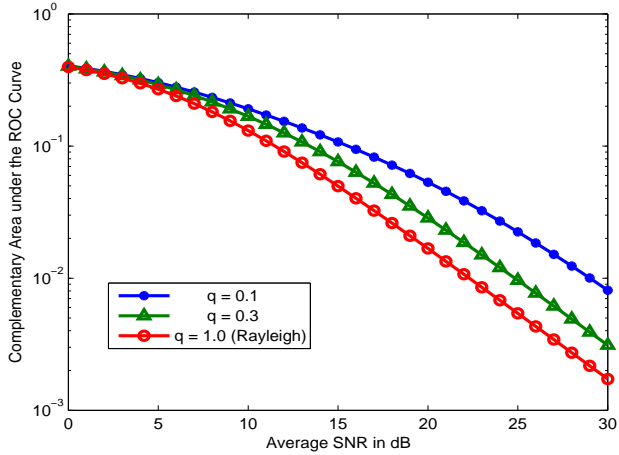} 
\caption{ $\overline{A'}_{\rm Hoyt}$ vs $\overline{\gamma}$ for i.i.d Hoyt (Nakagami${-}q$) fading channels with $u = 5$ and different values of $q$. } 
\end{figure} 
 On the contrary, the overall performance is rather poor in the low SNR regime $\overline{\gamma} < 0$dB, which indicates that the number of samples need to be increased.  Regarding the moderate SNR levels, it is shown that the value of $q$  affects the performance of the detector as, for example, the difference of $\overline{A}_{\rm Hoyt}$  between $q = 0.1$ and $q = 0.3$ is $6 \%$ for $\overline{\gamma} = 10$dB and $5 \%$  for $\overline{\gamma} = 20$dB.  
In the same context, Fig. $2$ illustrates the behaviour of the complementary area under ROC as a function of the average SNR. As in Fig. $1$, the overall performance of the behaviour increases as SNR increases while the effect of the value of $q$ is rather noticeable for moderate   SNR values. Indicatively,  the difference of $\overline{A}_{\rm Hoyt}$  between $q = 0.1$ and $q = 1.0$ is $80 \%$ for $\overline{\gamma} = 15$dB and $85 \%$ for $\overline{\gamma} = 25$dB.

\section{Conclusion}

This work analyzed the performance of energy detection over Hoyt (Nakagami${-}q$) fading channels. This was carried out by means of the area under the ROC curve and the complementary area under the ROC curve metrics which are single parameter performance measures. Novel analytic expression were  derived for the case of integer and arbitrary  values of the time bandwidth product. These expressions have a relatively simple algebraic form which renders them convenient to handle both analytically and numerically. To this end, they were employed in quantifying the effect of enriched multipath fading conditions in energy detection based spectrum sensing in the context of cognitive radio and radar systems and it was shown that the corresponding severity of fading affects the detector's performance, particularly for moderate SNR values.


\bibliographystyle{IEEEtran}

\end{document}